\documentclass[12pt]{article}
\usepackage[font=footnotesize,labelfont=bf,indent]{caption} 
\usepackage{subcaption}
\usepackage{setspace}
\usepackage{graphicx}
\usepackage{svg}
\usepackage[english]{babel}
\usepackage{enumitem} 
\usepackage[style=apa, backend=biber]{biblatex}
\usepackage{xifthen} 

\usepackage{xifthen} 
\usepackage{xr-hyper} 
\usepackage{hyperref} 
\usepackage{relsize} 
\usepackage{amsmath} 
\usepackage{amssymb} 
\usepackage{amsthm}  
\usepackage{mathrsfs} 
\usepackage{bm} 
\usepackage{algorithm,algorithmicx,algpseudocode} 

\newcommand{\tr}{^\intercal}
\newcommand{\mat}[1]{\bm{\mathrm{#1}}}
\newcommand{\tmat}[1]{%
	\ifthenelse{\equal{#1}{W}}{\bm{{\widetilde{\mathrm{#1}}}}}%
	{\bm{{\tilde{\mathrm{#1}}}}}%
}
\newcommand{\hmat}[1]{\bm{{\hat{\mathrm{#1}}}}}
\newcommand{\argmin}{\operatorname*{arg\,min}}

\newcommand{\setT}{\mathscr T}
\newcommand{\setD}{\mathscr D}
\newcommand{\setF}{\mathscr F}
\newcommand{\E}{\mathbb{E}}
\newcommand{\V}{\operatorname{Var}}
\newcommand{\bc}{\beta_c}
\newcommand{\bhc}{\hat\beta_c}
\newcommand{\btc}{\tilde \beta_c}

\newcommand{\bs}{\beta_{\star}}
\newcommand{\bhs}{\hat \beta_{\star}}
\newcommand{\bts}{\tilde \beta_{\star}}

\newcommand{\sumc}{\sum_{i \in \mathcal C}}
\newcommand{\sums}{\sum_{i \notin \mathcal C}}
\newcommand{\loss}{\mathcal O}
\renewcommand{\l}{\lambda}
\renewcommand{\P}[1]{\mathbb P\left\{#1\right\}}

\newcommand{\grad}[1]{%
  \raisebox{-2pt}{$\mathlarger{\mathlarger{\mathlarger{\nabla}}}\!\!_{#1}$}%
}
\newcommand{\minimaset}[2][]{%
  \ifthenelse{\isempty{#1}}%
    {\mathcal{B}^{#2}_{\lambda}}%
    {\mathcal{B}^{#2}_{#1}}%
}
\newcommand{\maminimaset}[2][]{%
  \ifthenelse{\isempty{#1}}%
    {\check{\mathcal{B}}^{#2}_{\lambda}}%
    {\check{\mathcal{B}}^{#2}_{#1}}%
}

\newcommand{\pense}{\texttt{pense}}
\newcommand{\Cor}{\operatorname{Cor}}

\makeatletter
\newcommand*{\addFileDependency}[1]{
\typeout{(#1)}
\@addtofilelist{#1}
\IfFileExists{#1}{}{\typeout{No file #1.}}
}\makeatother

\newcommand*{\myexternaldocument}[1]{%
  \externaldocument{#1}%
  \addFileDependency{#1.tex}%
  \addFileDependency{#1.aux}%
}

\myexternaldocument{supplemental}

\newcommand{\blind}{0} 

\algnewcommand\algorithmicinput{\textbf{Input:}}
\algnewcommand\Input{\item[\algorithmicinput]}
\newtheorem{proposition}{Proposition}

\title{}
\author{Siqi Wei, David Kepplinger}
\date{September 2022}

\def\spacingset#1{\renewcommand{\baselinestretch}%
    {#1}\small\normalsize} \spacingset{1}

\begin{document}

\if0\blind
{
    \title{\bf Stable and Robust Hyper-Parameter Selection Via Robust Information Sharing Cross-Validation}
    \author{David Kepplinger\thanks{
    This project was supported by resources provided by the Office of Research Computing at George Mason University (URL: https://orc.gmu.edu) and funded in part by grants from the National Science Foundation (Award Number 2018631).}    \hspace{.2cm}\\
    and \\
    Siqi Wei \\
    Department of Statistics, George Mason University}
    \maketitle
}
\fi

\if1\blind
{
    \bigskip
    \bigskip
    \bigskip
    \begin{center}
        {\LARGE\bf Information Sharing for Robust and Stable Cross-Validation}
    \end{center}
    \medskip
} \fi

\bigskip


\begin{abstract}
Robust estimators for linear regression require non-convex objective functions to shield against adverse affects of outliers.
This non-convexity brings challenges, particularly when combined with penalization in high-dimensional settings.
Selecting hyper-parameters for the penalty based on a finite sample is a critical task.
In practice, cross-validation (CV) is the prevalent strategy with good performance for convex estimators.
Applied with robust estimators, however, CV often gives sub-par results due to the interplay between multiple local minima and the penalty.
The best local minimum attained on the full training data may not be the minimum with the desired statistical properties.
Furthermore, there may be a mismatch between this minimum and the minima attained in the CV folds.
This paper introduces a novel adaptive CV strategy that tracks multiple minima for each combination of hyper-parameters and subsets of the data.
A matching scheme is presented for correctly evaluating minima computed on the full training data using the best-matching minima from the CV folds.
It is shown that the proposed strategy reduces the variability of the estimated performance metric, leads to smoother CV curves, and therefore substantially increases the reliability and utility of robust penalized estimators.
\end{abstract}

\noindent%
\vfill

\newpage
\spacingset{1.8} 

\section{Introduction}

In this paper we revisit a critical issue for applying robust penalized estimators: how to reliably select hyper-parameters of the penalty function.
In practice, cross-validation (CV) is by far the most prevalent strategy used to select hyper-parameters.
Besides certain adjustments of the CV sampling scheme, (e.g., stratified CV), performing multiple replications of CV, or using different evaluation metrics, the general procedure is almost always the same.
While computations can be burdensome CV has become a ubiquitous tool in any statistical learning framework.
Recent advances in asymptotic results for K-fold CV \parencite[][e.g.,]{austern_asymptotics_2020,bates_cross-validation_2024,li_asymptotics_2023} are further underlining the advantages of CV which were previously noticed only empirically.
While these theoretical guarantees for CV do not apply when used for penalized robust estimators with the LASSO or Elastic Net (EN) penalties, some empirical studies suggest good out-of-sample accuracy and variable selection can be achieved by utilizing CV with robust measures of the prediction accuracy \parencite{ronchetti_robust_1997,cohen_freue_robust_2019,khan_robust_2007,maronna_robust_2011,filzmoser_robust_2021,loh_scale_2021,monti_sparse_2021,sun_adaptive_2019,amato_penalised_2021}.
Reproducing these benefits in practical applications, however, is difficult because CV for robust penalized estimators tends to be highly unstable \parencite{she_gaining_2021,kepplinger_robust_2023,}, particularly in the presence of outliers in the response and contamination in the predictors.
Even if there is only measurement errors naïve leave-one-out CV with the non-robust LASSO fails \textcite{datta_note_2019}, and
issues tend to be much more severe when allowing for arbitrary contamination and using non-convex estimators.

Practically these issues manifest in very different answers for different random CV splits of the data.
These differences are often substantial and affect both the estimate of the prediction accuracy and hyper-parameters selection, leading to questionable results.
In the following we will illuminate the issues underlying the instability of CV for robust penalized estimators.
We further propose a novel CV strategy, called Robust Information Sharing (RIS) CV, which remedies these issues and provides faster, more reliable and stable estimation of prediction accuracy and thus hyper-parameter selection.

We focus on robust penalized estimators for the linear regression model,
\begin{equation}\label{eqn:lin-reg-model}
y_i = \mat x_i^\intercal \mat\beta_0 + \varepsilon_i, \quad i = 1,\dots,n,
\end{equation}
where the $p$-dimensional predictors $\mat x_i$ take on real values and the errors $\varepsilon_i$ are i.i.d.\ following an arbitrary symmetric distribution.
Given $n$ observations the primary goals are to predict out-of-sample responses for a new observation $\mat x^*$ and identify non-zero coefficients in $\mat\beta_0$.

A substantial body of literature on penalized regression estimators is concerned with indirect, fit-based criteria for model selection, like AIC and BIC.
These fit-based criteria have also been adapted for robust estimation, e.g., the robust BIC criterion \parencite{alfons_sparse_2013} or the predictive information criterion \parencite{she_gaining_2021}.
In practice, however, CV is still the dominating strategy for various reasons, including because it allows for comparisons between different estimators and because it does not rely on a robust estimate of the scale, which is itself a very challenging problem in high dimensions \parencite{maronna_correcting_2010,reid_study_2016,loh_scale_2021,dicker_variance_2014,fan_variance_2012}.

Considering a potentially large number of predictors, $p$, and less than half of the observations deviating from the model~\ref{eqn:lin-reg-model}, we are interested in selecting hyper-parameters for robust penalized regression estimators.
In this work we focus on elastic-net (EN) penalized M- and S-estimators (PENSEM and PENSE, respectively) as proposed in \textcite{cohen_freue_robust_2019}, defined as the minimizer of the objective function
\begin{equation}\label{eqn:objf}
\loss(\mat y - \mat X \mat\beta; \lambda, \alpha) := \ell(\mat y - \mat X \mat\beta ) + \lambda \sum_{j=1}^p \left(\frac{1-\alpha}{2}\beta_j^2 + \alpha |\beta_j| \right),
\end{equation}
and their extensions to the adaptive EN penalty \parencite[adaptive PENSE/PENSEM;][]{kepplinger_robust_2023-1}.
To apply these estimators successfully in practice, appropriate values for the hyper-parameters $\lambda > 0$ and $\alpha \in [0, 1]$, which govern the strength and form of the EN penalty, must be chosen in a data-driven fashion.
The exact value of $\alpha$ is typically less critical to the prediction accuracy of the estimate than the strength of the penalization, $\lambda$.
We will therefore focus on selecting $\lambda$ to simplify the exposition.

Robustness in~\eqref{eqn:objf}, and hence stability and reliability in the presence of outliers in the response and contamination in the predictors, is achieved through a robust loss function, $\ell$.
The loss for the robust penalized M-estimator is $\ell_M(\mat r) = \frac{1}{2 n} \sum_{i=1}^n \rho(r_i / s)$, with a pre-determined scale of the error term, $s > 0$.
For the penalized S-estimator, the loss function is the M-scale of the residuals, $\ell_S(\mat r) = \frac{1}{2} \sigma^2_M(\mat r)$, defined implicitly by the equation
\begin{equation}\label{eqn:mscale}
\delta = \frac{1}{n} \sum_{i=1}^n \rho \left( \frac{r_i}{\sigma^2_M(\mat r)} \right).
\end{equation}
For both the M- and S-loss a bounded and hence non-convex $\rho$ function is necessary to achieve high robustness towards arbitrarily contaminated data points.
\textcite{cohen_freue_robust_2019} showed that $\delta$ in~\eqref{eqn:mscale} is the finite-sample breakdown point of (adaptive) PENSE, i.e., the proportion of observations that can be arbitrarily without leading to an infinitely biased  estimate.
The $\rho$ function is usually chosen to behave like the square function around 0 and smoothly transition to a constant beyond a certain cutoff value.
Typical examples are Tukey's bisquare or the LQQ function \parencite{koller_sharpening_2011}.
The boundedness is necessary to achieve the desired robustness, but it leads to a non-convex objective function.

The primary objective in this paper is to select the overall penalization level, $\lambda$, for robust penalized estimators.
The main contribution is two-fold.
First, in Section~\ref{sec:issues}, we illuminate and investigate the major drivers causing issues in applying CV to robust penalized estimators, i.e., the combination of a non-convex objective function with local minima determined by outliers and the penalty function.
Second, we propose a new, robust and reliable cross-validation strategy, called Robust Information Sharing CV (RIS-CV), that simultaneously tackles these issues while also reducing the computation burden compared to regular CV (N-CV).
While RIS-CV is applicable to any robust penalized estimator, we will use (adaptive) PENSE for illustration purposes throughout most of this paper.

\subsection{Notation}

The index set of the complete training data of $n$ observations is denoted as $\setT = \{1, \dotsc, n \}$.
The subsets of the training data used to estimate the parameters in the $K$ CV folds are denoted by $\setF_1, \dotsc, \setF_K$, $\setF_k \subset \setT$.
The set of minima for a given data set $\setD$ and penalty parameter $\lambda$ is denoted by $\minimaset{\setD} = \left\{ \mat\beta \colon \grad{\mat\beta} \loss(\mat\beta; \setD, \lambda) = \mat 0 \right\}$.
We assume that only the $M \geq 1$ best minima are retained and that the minima are ordered by the value of the loss function, i.e., in a set of minima $\minimaset{\setD} = \{ \mat\beta_1, \dotsc, \mat\beta_M \}$, $\loss(\mat\beta_1; \setD, \lambda) \leq \loss(\mat\beta_2; \setD, \lambda) \leq \cdots \leq \loss(\mat\beta_M; \setD, \lambda)$.
If the data set on which the objective function is being evaluated is obvious from the context, $\setD$ will be omitted from the notation.

None of the non-convex optimization routines employed in this paper can guarantee to find the actual global minimum.
Any notion of a ``global'' minimum is therefore to be understood as the local minimum with the smallest value of the objective function, among all local minima uncovered by the non-convex optimization routine.
We will denote this presumptive global minimum, i.e., the first of minimum in $\minimaset{\setD}$, by $\mat\beta^*(\l; \setD)$.

\section{Cross-validation for Robust Penalized Estimators}\label{sec:issues}

Before we shed light on why standard, or ``naïve'', CV (N-CV) fails for robust penalized estimators, we give a brief review of standard N-CV commonly used for robust and non-robust penalized regression estimators.

For N-CV we first compute the global minimizers of \eqref{eqn:objf} over a fine grid of $\l$ values, $\mathcal L = \{ \l_1, \dotsc, \l_K\}$, $\l_1 > \l_2 > \cdots > \l_q > 0$, using all available observations, $\setT$.
To estimate the prediction accuracy of these minimizers, N-CV randomly splits the training data $\setT$ into $K$ approximately equally sized subsets, or folds, $\setF_1, \dotsc, \setF_K$ with $\setF_k \subset \setT$ such that $\bigcup_k \setF_k = \setT$ and $\setF_k \cap \setF_{k'} = \emptyset$ for all $k \neq k'$.
On each of these CV folds, the global minimizers of \eqref{eqn:objf} are computed using only observations in $\setT \setminus \setF_k$ using the same penalty grid $\mathcal L$ as for the complete training data.
Denoting these global minimizers by $\hmat\beta_{k,\l}$, we compute the prediction errors on the left-out observations as $e_{i,\l} = y_i - \mat x_i\tr \hmat\beta_{k,\l}$,$i \in \setF_k$.
The prediction accuracy of the estimate $\hmat\beta_\l$ is then estimated for each $\l \in \mathcal L$ by summarizing the prediction errors, usually using a measure of the scale of these prediction errors.
We will denote that measure of prediction accuracy as $S(\l)$.
The prevalent choice for $S$ is the root mean-square prediction error (RMSPE),
$$
\widehat{\text{RMSPE}}(\l) = \sqrt{\frac{1}{n}\sum_{i=1}^n e_{i, \l}^2}.
$$
In the potential presence of outliers, however, it is commonly argued \parencite{cohen_freue_robust_2019,kepplinger_robust_2023-1,smucler_robust_2017,she_gaining_2021} that robust estimators of the prediction accuracy should be used since outliers are not expected to be well predicted by the model.
Common choices are the median absolute prediction error (MAPE) or the $\tau$-size \parencite{yohai_high_1988} of the prediction errors, given by
\begin{align*}
\widehat{\text{MAPE}}(\l) &= \operatorname*{Median}_{i = 1, \dotsc, n} |e_{i, \l}|,\\
\hat\tau(\l) &= \widehat{\text{MAPE}}(\l) \sqrt{
  \frac{1}{n} \sum_{i=1}^n \min \left( c_\tau,
    \frac{ | e_{i,\l} | }{ \widehat{\text{MAPE}}(\l) }
  \right)^2},
  \quad c_\tau > 0.
\end{align*}

To reduce the Monte Carlo error incurred by a single CV split and to obtain a rough estimate of the variance of the error measure, CV can be repeated with different random splits.
With $R$ replications of K-fold N-CV, the estimated prediction accuracy using metric $S$ is
\begin{align*}
\hat S(\l) &= \frac{1}{R} \sum_{r=1}^R \hat S^{(r)}(\l),&
\text{SE}(\hat S(\l)) &= \sqrt{\frac{1}{R} \sum_{r=1}^R \left( \hat S^{(r)}(\l) - \hat S(\l) \right)^2}.
\end{align*}

The penalty level $l$ is then chosen as the one leading to the smallest measure of prediction error, $\hat\l = \min_{\l \in L} \hat S(\l)$.
Alternatively, $\l$ can be chosen considering the estimated standard errors, e.g., using th ``1-SE-rule''.
In practice, more utility lies in a plot of the prediction accuracy against the penalization strength.
The ``CV curve'' plots both the estimated prediction accuracy and the estimated standard error against the penalization strength or the $L_1$ norm of the estimates at the different $\l \in \mathcal L$.
The practitioner can then choose the hyper-parameter leading to the best prediction accuracy or may choose a hyper-parameter that better balances model complexity with prediction accuracy.

Figure~\ref{fig:appl-cav-ind-select} shows this CV curve for a classical (least-squares based) EN estimator (left) and the robust adaptive PENSE estimator (right) in the real-world application from Section~\ref{sec:appl-cav}.
It is obvious that the CV curve provides valuable insights into the effects of penalization and overall prediction accuracy of the the models on the penalization path, with two striking observations.
First, the classical adaptive EN estimator does not seem to perform particularly well in this example, with the intercept-only model (at the far right with the highest considered penalization strength) yielding almost as good as a prediction accuracy as the less sparse estimates.
Clearly, a practitioner may question the applicability of the EN estimator or the linear regression model in this case.
Second, the CV curve for adaptive PENSE is highly non-smooth.
While adaptive PENSE seems to find models with better prediction accuracy than the intercept-only model, the standard errors are quite large and the highly irregular shape of the CV curve does not instill much confidence in those estimates.
We will see in Section~\ref{sec:empirical-studies} that the poor performance of classical EN estimator is due to outliers in the training data.
But the more important question is why do small changes in the penalization level lead to such drastic changes in the estimated prediction accuracy for adaptive PENSE?

\begin{figure}[t]
    \centering
    \begin{subfigure}[t]{0.44\textwidth}
        \centering
        \includegraphics[width=\textwidth]{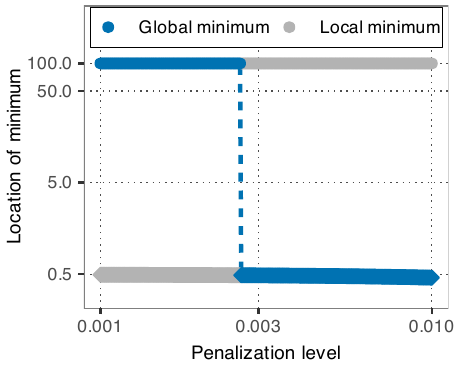}
        \caption{%
        }
        \label{fig:nonsmooth-demo}
    \end{subfigure}
    \hfill
    \begin{subfigure}[t]{0.55\textwidth}
        \centering
        \includegraphics[width=\textwidth]{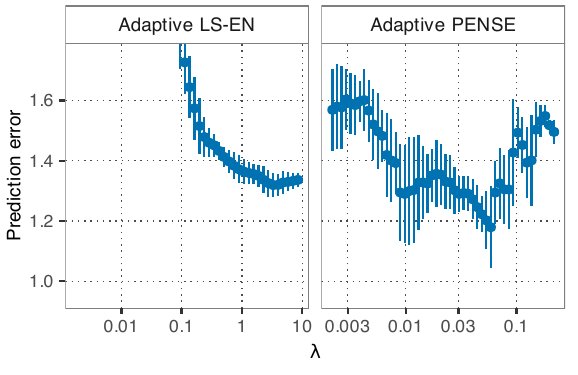}
        \caption{%
        }
        \label{fig:appl-cav-ind-select}
    \end{subfigure}
    \caption{%
    Demonstration of non-smooth penalization paths (a) and cross-validation curve (b). 
    Sub-figure (a) shows the location of two local minima (gray dots) and the non-smootheness of the global minimum (blue dots) in a simple scenario explained in Section~\ref{sec:suppl-nonsmooth-example} of the supplementary materials.
    The CV curves in (b) show the failings of classical adaptive EN estimate and non-smoothness of N-CV with the robust adaptive PENSE estimate in the CAV study from Section~\ref{sec:appl-cav}.
    }
\end{figure}

\subsection{Failings of N-CV for Penalized Robust Estimators}

In short, the blame for the non-smoothness of the PENSE N-CV curve is on the non-convexity of the objective function combined with the presence of outliers and contamination.
Below we will highlight the issues that the non-convexity and outliers create for CV of robust penalized estimators.
Together, these issues may lead to an undesirable coefficient estimate and a cross-validated prediction accuracy that is substantially biased, has high variance, or both.

\paragraph{Non-smoothness of the penalization path.}
Due to the non-convexity of the objective function~\eqref{eqn:objf}, it has in general more than one minimum.
PENSE and other robust penalized estimators are usually defined as the global minimizer of this objective function.
While the objective function is smooth in $\l$, the global minimum is not necessarily so.
In Proposition~\ref{prop:suppl-nonsmooth} in the Supplementary material we show that in the simple univariate setting and under certain conditions we can find at least one $\l$ where the path of the global minimum of a penalized M-estimator, denoted by $\hat\beta^*(\l)$ has a discontinuity, i.e.,
$
\lim_{\delta \to 0} | \hat\beta^*(\l - \delta) - \hat\beta^*(\l + \delta) | > 0.
$
We further provide a simple example setting where these conditions are satisfied with non-zero probability.
Figure~\ref{fig:nonsmooth-demo} shows an instance of this example where the objective function has two minima, one around 0.5 and one around 100, for all $\l \in \mathcal L$.
The global minimum, however, has a discontinuity around $\l = 0.003$, jumping from somewhere around 100 for small $\l$ values to somewhere around 0.5.
While this example shows a univariate penalized M-estimator, similar behavior can be found for PENSE and other robust penalized estimators, as well as in high-dimensional problems.
In these instances, the non-smoothness of the global minimum is even more pronounced (the S-loss tends to have more minima than the M-loss and more covariates tend to lead to more minima).

The primary reason to focus on the global minimizer of the objective function is that it may possess many desirable statistical properties.
The global minimum of the PENSE and adaptive PENSE objective function, for instance, is root-n consistent, finite-sample robust, and for adaptive PENSE possesses the oracle property \parencite{kepplinger_robust_2023-1}.
These properties, however, only pertain to the global minimum at a properly chosen penalty parameter $\l$.
In practice, when many different $\l$ values must be tried, considering only the global minimum is fallacious.
In the example scenario from Figure~\ref{fig:nonsmooth-demo}, the majority of the data follows a linear regression model with $\beta_\star = 100$, and hence an estimate close to that would lead to a reasonably small RMSPE for all $\l$.
If we were to consider only the global minimum (blue points), however, a stark difference would be seen between prediction accuracy with small versus large $\l$ values.
In situations where a larger $\l$ value and hence a sparser solution may be preferable, the focus on the global minimum would preclude selection of a good estimate.

\paragraph{Mismatch between the global minima and the N-CV solutions.}
A related problem arises when estimating the prediction accuracy of the minima of the objective function using N-CV.
The objective function evaluated on a random subset of the data, e.g., on the CV training data $\setD_k = \setT \setminus \setF_k$, may also possess multiple minima.
If the global minimum on $\setD_k$ describes the same signal as the global minimum on the full training data, $\setT$, however, is unknown.
Referring again to the example in Figure~\ref{fig:nonsmooth-demo}, if the global minimum on $\setD_k$ for a small $\l=0.001$ is around $0.5$, it would bear little information for estimating the prediction accuracy of the global minimum on $\setT$, which is around 100.
The more local minima, the greater the chances that the global minimum on $\setD_k$ is unrelated to the global minimum on $\setT$.
Among the $K$ CV folds some may yield a related global minimum, while others do not.
With N-CV, however, the prediction accuracy is estimated from both the related and unrelated minima, potentially introducing  substantial bias and hence leading to nonsensical results.
Using robust measures of prediction accuracy for N-CV, like the MAPE or the $\tau$-size, also does not solve this problem.
The unrelated minima could give prediction errors which appear as outliers and combined with true outliers in the data could outnumber the useful prediction errors and hence break the robust measures and lead to unbounded bias.
Repeating N-CV often enough usually helps to smooth-out these effects, but depending on the severity of the mismatches, a large number of replications may be necessary which creates computational problems and the actual variance of the prediction accuracy may be overestimated.

\section{Robust Information Sharing CV}\label{sec:ris-cv}

We now present Robust Information Sharing Cross-Validation (RIS-CV), combining three remedies to overcome the issues posed by non-convex penalized loss functions and outliers.

\paragraph{Remedy 1: tracking multiple minima.}
These first remedy is to keep track of multiple minima over the entire penalization grid, $\mathcal L$, for each data set $\setD$ (i.e., the full training data and all the CV training data sets).
For every $\l \in \mathcal L$ we retain $M^{(\setD)}_\l$ unique minima, denoted by $\minimaset{\setD}$.
To keep the computational complexity at bay, we limit the maximum number of unique minima to $M$, i.e., $M^{(\setD)}_\l \leq M$.
In the numerical experiments below we set $M=40$.

We have shown in Section~\ref{sec:issues} that the regularization path may not be smooth if $M = 1$.
It is also easy to see that relaxing the penalty monotonically increases the number of local minima.
Moreover, if $\hmat\beta(\l) \in \minimaset{\setD}$ is a minimum of the objective function for $\l$, then for any $\delta > 0$ there exists an $\epsilon > 0$ such that there is a minimum $\hmat\beta(\lambda - \epsilon)$ with $\| \hmat\beta(\l - \epsilon) - \hmat\beta(\l) \| < \delta$.
While this may not be the global minimum of the objective function at $\l - \epsilon$, tracking multiple minima increases the chances that $\hmat\beta(\l - \epsilon) \in \minimaset[\lambda - \epsilon]{\setD}$.
Instead of a single, non-smooth regularization path, with RIS-CV we hence capture multiple smooth regularization paths.
Tracking local minima is sometimes exploited for computational reasons \parencite{kepplinger_robust_2023-1,alfons_robusthd_2016}, but these software implementations nevertheless return and utilize only the global minima for estimation.
In RIS-CV, on the other hand, we harness these multiple smooth regularization paths to improve the reliability of the regularization path and the estimated CV curve.

\paragraph{Remedy 2: reducing the search space.}
Since the objective function is non-convex, the minima uncovered by numerical algorithms depends on the starting point provided to that numerical procedure.
The quality of the best minimum is thus tied to the quality of the starting points.
In particular, robustness properties heavily depend on the starting points being unaffected by outliers and contamination.
Robust estimators typically utilize a semi-guided or random search for outlier-free subsets of the data, computing classical least-squares-based estimate on these subsets to use as starting points.
These strategies are computationally taxing because they involve computing many estimates and/or projections of the data, and computing robust scale estimates many sets of residuals \parencite{alfons_sparse_2013,kepplinger_robust_2023-1}.

Given the goal is to estimate the prediction performance of minima in the full training data, $\setT$, however, it is not necessary to locate all minima (or the global minimum) of the objective function on the CV training data $\setD_k = \setT \setminus \setF_k$.
Only the local minima in the vicinity of the minima in $\minimaset{\setT}$ are actually needed.
Therefore, in RIS-CV we use the minima in $\minimaset{\setT}$ as starting points for the numerical optimization in all CV folds.
This effectively restricts the search space to the region of interest which leads to two major benefits over re-computing starting points using the same costly procedure utilized for the training data $\setT$: (i) faster computation and (ii) increased chance of each minimum in $\minimaset{\setT}$ having a corresponding, related minimum in $\minimaset{\setF_k}$.

\paragraph{Remedy 3: matching minima based on similarity.}
The final remedy is to match the minima in $\minimaset{\setT}$ with their most related minimum in each CV fold, $\minimaset{\setF_k}$.
We propose to measure relatedness of minima based on the similarity of the robustness weights associated with these minima, harnessing that for any $\mat\beta$ and data set $\setD$ the penalized M- and S-loss can be re-cast as a weighted penalized least-squares loss.
For example, using the weights
\begin{equation}\label{eqn:method-weights}
w_i(\mat\beta) = \frac{\rho'(\tilde r_i(\mat\beta)) / \tilde r_i(\mat\beta)}{
  \sum_{k \in \setD} \rho'(\tilde r_k(\mat\beta)) \tilde r_k(\mat\beta) },
\quad
i \in \setD,
\end{equation}
the PENSE objective function can be re-cast as
$$
\loss(\mat\beta; \setD, \lambda) =
  \frac{1}{2 |\setD|} \sum_{i \in \setD} w_i(\mat\beta) (y_i - \mat x_i\tr \mat\beta)^2 +
  \lambda P(\mat\beta),
$$
where $\tilde r_i(\mat\beta) = (y_i - \mat x_i\tr \mat\beta) / \hat\sigma_M(\mat r(\mat\beta))$ are the residuals scaled by their M-scale estimate.

The weights encode the ``inlyingness'' of an observation relative to the regression hyperplane spanned by $\mat\beta$.
Observations close to the hyperplane get larger weight, while observations far away and hence outlying get a weight of 0.
For RIS-CV we define the similarity between two coefficient vectors, $\omega(\mat\beta_1, \mat\beta_2; \setD)$, as the Pearson correlation between the corresponding weight vectors $\mat w(\mat\beta_1)$ and $\mat w(\mat\beta_2)$,
\begin{equation}\label{eqn:method-similarity}
\omega(\mat\beta_1, \mat\beta_2; \setD) = 
  \frac{\frac{1}{|\setD|} \sum_{i \in \setD} w_i(\mat\beta_1) w_i(\mat\beta_2) - 
        \overline{w}(\mat\beta_1) \overline{w}(\mat\beta_2)}
       {\sqrt{ \left[ \frac{1}{|\setD|} \sum_{i \in \setD} w_i(\mat\beta_1)^2 - \overline{w}(\mat\beta_1)^2 \right]
               \left[ \frac{1}{|\setD|} \sum_{i \in \setD} w_i(\mat\beta_2)^2 - \overline{w}(\mat\beta_2)^2 \right]
        }},
\end{equation}
with $\overline{w}(\mat\beta) = \frac{1}{|\setD|} \sum_{i \in \setD} w_i(\mat\beta)$.

We utilize the weight-similarity $\omega(\mat\beta_1, \mat\beta_2; \setD)$ to match minima in $\minimaset{\setT}$ with their closest counterpart in each CV fold, $\minimaset{\setF_k}$.
Specifically, for a set of minima $\minimaset{\setT}$ we define its collection of CV-surrogates from fold $k = 1, \dotsc, K$ as
\begin{equation}\label{eqn:method-surrogate-set}
\maminimaset{\setF_k} = \left\{
  \argmin_{\mat\beta^* \in \minimaset{\setF_k}} \omega(\mat\beta^*, \mat\beta_j; \setF_k)
  \colon
  j = 1, \dotsc, M^{\setT}_\lambda, \mat\beta_j \in \minimaset{\setT}
\right\}.
\end{equation}
Hence the $q$-th element in $\maminimaset{\setF_k}$ is the minimum from CV fold $\setF_k$ most similar to the $q$-th element in $\minimaset{\setT}$.
There can be duplicates in $\maminimaset{\setF_k}$.

Matching minima from the complete training data to their closest counterparts in each CV fold allows us to more reliably estimate the prediction accuracy of each minimum than using merely the ordering of the minima based on their value of the objective function.
With a non-convex loss function the minimum with lowest objective value in $\minimaset{\setT}$ may capture a very different signal than the minimum with lowest objective value in $\minimaset{\setF_k}$.
In contrast, our strategy matches each minimum in $\minimaset{\setT}$ with a minimum in $\minimaset{\setF_k}$ that best agrees on the outlyingness of the observations in CV fold $\setF_k$.
While we cannot guarantee that they actually represent the same signal, the matching is likely more informative than the order of the minima, which is also supported by our empirical results below.

The proposed weight-based similarity has several advantages over distances between the coefficient estimates or residuals.
First, weight-similarity is dimensionless and does not depend on the number of covariates, their covariance structure or the scale of the response.
Second, observations that are deemed outliers by both estimates do not affect the weight-similarity, whereas measures directly utilizing the residuals can be arbitrarily affected by outliers.
Third, outliers are often the drivers behind local minima; minima agreeing on the outlyingness of observations thus likely describe a comparable signal.

Once the CV-surrogates from each CV fold $k$ and penalty parameter $\l \in \mathcal L$ are determined, RIS-CV utilizes the weights~\eqref{eqn:method-weights} to quantify the prediction accuracy of every minima $\hmat\beta_q \in \minimaset{\setT}$.
The prediction accuracy is estimated by a weighted standard deviation of the CV prediction errors, weighted by the outlyingness of each observation as estimated on the complete training data:
\begin{equation}\label{eqn:method-weighted-rmspe}
    \hat E_{\lambda,q} = \sqrt{\frac{1}{ \sum_{i \in \setT} w_i(\hmat\beta_q)} \sum_{k = 1}^K 
        \sum_{i \in \mathcal T\setminus \setF_k} w_i(\hmat\beta_q) \left(y_i - \mat x_i\tr \hmat\beta^k_q \right)^2},
    \quad q = 1, \dotsc, |\minimaset{\mathcal T}|,
\end{equation}
where $\hmat\beta^k_q$ is the $q$-th minimum in $\maminimaset{\setF_k}$.
This effectively ignores the prediction error of observations that are deemed outliers by that particular $\hmat\beta_q$.
Outliers cannot be expected to be predicted well by the model, hence the prediction errors for outliers should not affect the overall assessment of an estimate's prediction accuracy.
Moreover, by the definition of the M-scale~\eqref{eqn:mscale}, an S-estimator cannot assign a zero weight to more than $\lfloor \delta n \rfloor$ observations.
Therefore, as long as there are less than $\lfloor \delta n \rfloor$ outliers in the training data, even if $\hmat\beta_q$ describes an illicit signal driven by these outliers, many non-outlying observations will have a weight greater than 0 and hence the prediction accuracy $\hat E_{\lambda,q}$ will reflect the poor prediction of non-outliers.
It will thus show up as a minimum with poor prediction accuracy.

Compared to the usual robustification of N-CV through robust measures of the prediction error, our approach connects the estimated prediction error more closely to the estimated outlyingness, reducing the risk of misrepresenting an estimate's prediction accuracy.
While robust measures such as MAPE and $\tau$-size guard against the effects of arbitrarily large prediction errors for a proportion of prediction errors, these measures do not discriminate whether the large prediction error comes from an observation that is determined to be an outlier or not.
Therefore, observations that may have a non-zero influence on the estimated coefficients (inliers) are allowed to have very large prediction error.
This disconnect between the outlyingness of observations and their effect on the estimated prediction error can lead to high variance in the estimated prediction accuracy.

Just like with N-CV, the prediction accuracy estimated by RIS-CV depends on the random CV splits $\setF_k$, $k=1,\dotsc,K$, and hence is a stochastic quantity.
RIS-CV should therefore also be repeated several times to assess the variability of the estimate, but the number of replications is usually much smaller than the number of replications needed for N-CV.
Moreover, since the starting points are the same in all replications and CV folds, RIS-CV can be computed much faster than N-CV.
We generally suggest to repeat RIS-CV 5--20 times, depending on the complexity of the problem.
More complex problems may lead to higher variability in the estimated prediction accuracy and require more RIS-CV replications.
The complete RIS-CV strategy with replications is detailed in Algorithm~\ref{alg:rc-procedure}.

For faster computations by leveraging the smoothness of the many different penalization paths, the \pense\ package performs step 1 in Algorithm~\ref{alg:rc-procedure} for all $\l \in \mathcal L$ before the replicated RIS-CV (steps 2--10) is applied.
The RIS-CV procedure yields estimates of the prediction accuracy for up to $M$ minima at each level of the penalty parameter, leading to several ways the ``optimal'' penalty parameter can be selected.
We simply select the minimum with best prediction accuracy for each $\lambda \in \mathcal L$, $\hat q_\lambda = \argmin_{q = 1, \dotsc, |\minimaset{\setT}|} \hat E_{\lambda,q}$.
These estimates can be used to build the RIS-CV curve which, alongside the associated standard errors, can be used to judge the model's suitability for the problem at hand and to select the optimal penalty parameter.
The standard errors could also be used to choose $\hat q_\lambda$ but the numerical experiments below do not suggest an improvement over the simpler strategy applied here.
The \pense\ package returns the estimated prediction accuracy and standard errors for all minima and hence allows the user to utilize more sophisticated strategies if desired.

\begin{algorithm}[ht]
\caption{Robust Information Sharing CV}
\label{alg:rc-procedure}
\begin{algorithmic}[1]
	\Input Standardized data set $\{(y_1, \mat x_1), \dotsc, (y_n, \mat x_n)\}$,
        fixed hyper-parameter $\lambda$,
        the number of folds $K$,
        the maximum number of minima retained $M$ and
        the number of cross-validation replications $R$.
	\State Compute up to $M$ unique minima using all observations $\setT = \{1, \dotsc, n \}$.
    This set of minima is denoted by $\minimaset{\mathcal T}$.
	\For{$r = 1, \dotsc, R$}
		\State Split the data into $K$ cross-validation folds, denoted by $\setF_1, \dotsc, \setF_K$, such that $\setF_k \cap \setF_{k'} = \emptyset$, for $k\neq k'$, and $\bigcup_{k=1}^K \setF_k = \setT$.
		\For {$k = 1, \dotsc, K$}
            \State Compute up to $M$ unique minima using the observations in $\setT \setminus \setF_k$, denoted by $\minimaset{\setF_k}$.
            The minima are computed by using $\minimaset{\mathcal T}$ as starting points for the non-convex optimization.
            \State From $\minimaset{\setF_k}$ determine the CV-surrogates $\maminimaset{\setF_k}$ according to~\eqref{eqn:method-surrogate-set}.
        \EndFor
        \State Estimate the robust weighted RMSPE for each minimum $\hmat\beta_q \in \minimaset{\mathcal T}$, $q = 1, \dotsc, |\minimaset{\mathcal T}|$ using~\eqref{eqn:method-weighted-rmspe}, denoted by $\hat E_{\lambda,q}^{(r)}$.
    \EndFor
    \State Compute the average robust weighted RMSPE and its standard error
    \begin{align*}
    \hat E_{\lambda,q} &= \frac{1}{R} \sum_{r=1}^R \hat E_{\lambda,q}^{(r)}, \\
    \widehat{\text{SD}}_{\lambda,q} &= \sqrt{\frac{1}{R-1} \sum_{r=1}^R (\hat E_{\lambda,q}^{(r)} - \hat E_{\lambda,q})^2}.
    \end{align*}
\end{algorithmic}
\end{algorithm}

\section{Empirical Studies}\label{sec:empirical-studies}

We demonstrate the advantages of RIS-CV over naïve CV in a real-world application and a simulation study.
Additional real-world applications and empirical results are presented in the supplementary materials.

\subsection{Real-World Application: Cardiac Allograft Vasculopathy}\label{sec:appl-cav}

In this real-world application we utilize the adaptive EN S-estimator (adaPENSE) to build a biomarker for cardiac allograft vasculopathy (CAV) from protein expression levels.
CAV, a common and often life-threatening complication after receiving a cardiac transplant, is characterized by the narrowing of vessels that supply oxygenated blood to the heart.
The usual clinical biomarker for CAV is the percent diameter stenosis of the left anterior descending artery.
The data is obtained from \textcite{kepplinger_robust_2023}, which use a synthetic replicate of the restricted original data in \textcite{cohen_freue_robust_2019}.
The goal in this application is to fit a linear regression model to the stenosis of the artery using expression levels of 81 protein groups from a total of $N=37$ patients.
The adaptive PENSE estimator is tuned to a breakdown point of 20\% and uses an EN penalty with $\alpha=0.75$.
For RIS-CV we retain up to 40 local solutions.

Figure~\ref{fig:appl-cav-ind} shows the CV curves (left panel) and the respective first-order changes (right panel) from two independent replications of 10-fold CV.
The gray curves show the $\tau$-size of the prediction error for varying penalty levels estimated by N-CV.
The blue lines, in contrast, show the weighted RMSPE estimated by RIS-CV.
It is clear that RIS-CV leads to more consistent results across the two replicates, and that the CV curves from RIS-CV are smoother than those estimated by N-CV.

When averaging five replications the differences between naïve CV and RIS-CV are even more obvious.
Figure~\ref{fig:appl-cav-cv}a shows the estimated prediction errors for adaptive PENSE using N-CV (left) and RIS-CV (right).
With N-CV it is difficult to identify an appropriate penalization level.
Values around $\lambda \approx 0.01$ and $\lambda \approx 0.06$ seem to give similar prediction accuracy, but the estimated prediction errors have high variance and the selected models for these respective penalization levels are quite different (18 vs.\ 7 non-zero coefficients).
RIS-CV, on the other hand, identifies a tight range of penalization levels around $\lambda \approx 0.02$ that seem to yield the best prediction accuracy in this data set.
The smoothness of the RIS-CV curve is clearly advantageous in identifying a good penalization level in this application.

We further analyze the smoothness of the solution path for the minima selected by N-CV and RIS-CV.
N-CV always selects the global minimum at each penalization level, while RIS-CV selects the minimum with the best estimated prediction accuracy.
Figure~\ref{fig:appl-cav-cv}b shows the $L_1$ norm of the coefficients on the solution path for N-CV and RIS-CV.
Due to the flexibility of selecting a non-global minimum, RIS-CV leads to a smoother solution path in terms of the $L_1$ norm of the coefficients.
In this example, N-CV exhibits two moderate discontinuities along the solution path, which could be one of the reasons why N-CV gives ambiguous results in the left panel of Figure~\ref{fig:appl-cav-cv}a.

In addition to the smoother CV curve and penalization path, RIS-CV is about 25\% faster to compute.
Particularly in applications with many local minima, locating starting points for the non-convex optimization of the adaptive PENSE from scratch can take considerable time.
By using only the local minima from the fit to the full data, RIS-CV drastically reduces the computational burden.

\begin{figure}[t]
  \centering
  \includegraphics[width=1\linewidth]{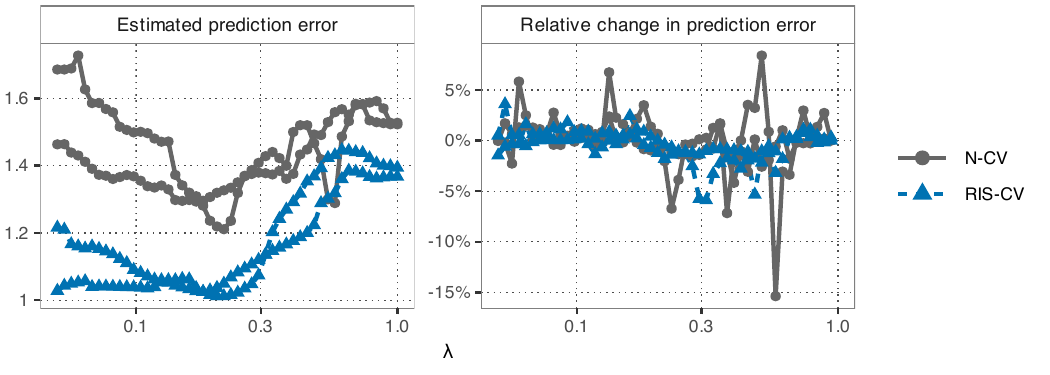}
  \caption{%
  Estimated prediction errors (left) and the respective changes (right) from two independent 10-fold CV runs for adaPENSE applied to the CAV data set.
  The blue curves represent the estimated weighted RMSPE from RIS-CV, while the gray curves show the $\tau$-size of the prediction errors estimated by N-CV.%
  }
  \label{fig:appl-cav-ind}
\end{figure}

\begin{figure}[t]
  \centering
  \includegraphics[width=1\linewidth]{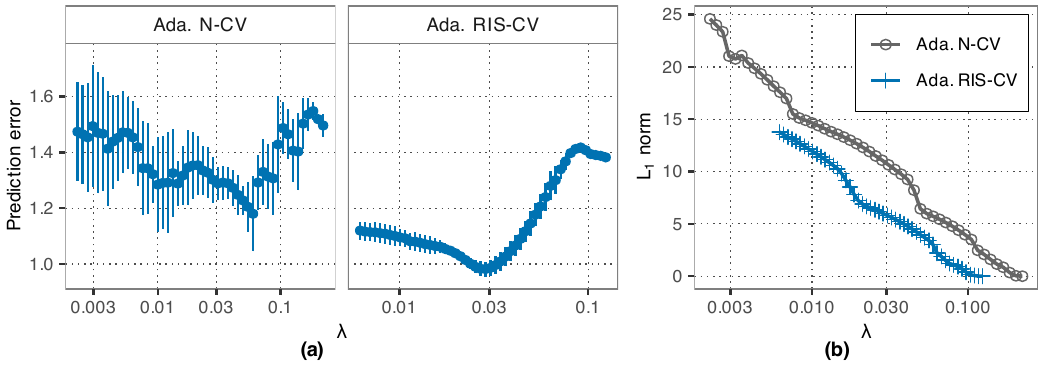}
  \caption{%
  Results for the CAV data set:
  (a) CV curves from 5 replications of 10-fold CV and
  (b) the $L_1$ norm (excluding the intercept) of the selected minimum of the adaPENSE objective function versus the penalization level ($\lambda$).
  In (a) the left panel (``Ada. N-CV'') shows the $\tau$-size of the prediction error of adaPENSE estimated by N-CV while the right panel (``Ada. RIS-CV'') shows the estimated weighted RMSPE of adaPENSE estimated by RIS-CV.%
  }
  \label{fig:appl-cav-cv}
\end{figure}

In Section~\ref{sec:suppl-additional-empirical-results} of the supplementary materials we present two additional real-world applications with similar conclusions as in the CAV study.
In these additional applications we can further compute the true prediction error on an independent test set.
The effects of the local minima in these applications are again noticeable but less pronounced than in the CAV study.
Nevertheless, RIS-CV leads to smoother CV curves and penalization paths, as well as approximately 3\% improvement in prediction accuracy and computational speed.

\subsection{Simulation Study}

In the real-world applications we demonstrate that RIS-CV leads to smoother CV curves and in turn to better selection of the penalty parameters.
We further show that the minimum selected by RIS-CV, which is not necessarily a global minimum, can lead to better out-of-sample prediction.
To underscore that these advantages also hold in other data configurations and data generating processes (DGP) we present the results from a simulation study.
Throughout this study we compare RIS-CV with N-CV for PENSE.
We consider a DGP similar to \textcite{kepplinger_robust_2023-1}.
The data is generated according to the model
\begin{equation}\label{eqn:sim-true-model}
y_i = \mat x_i\tr \mat\beta_0 + \varepsilon_i, \quad i = 1,\dotsc,n,
\end{equation}
where $\mat x_i$ is the $p$-dimensional covariate vector, $\mat\beta_0 = (1, \dotsc, 1, 0, \dotsc 0)\tr$ is the true coefficient vector with the first $s = \lfloor \log(n) \rfloor$ entries equal to 1 and the others are all 0.
The covariates $\mat x_i$ follow a multivariate t distribution with four degrees of freedom and AR(1) correlation structure, $\Cor(X_j, X_{j'}) = 0.5^{|j - j'|}$, $j,j' = 1, \dotsc, p$.
The i.i.d.\ errors, $\varepsilon_i$, follow a symmetric distribution $F$ with a scale chosen such that $\mat\beta_0$ explains about 50\% of the variation in $y_i$ (i.e., $\text{SNR} \approx 1$).
For this, the empirical variance in $\mat\varepsilon$ is measured by the empirical standard deviation if $F$ is Gaussian and by the $\tau$-size for other error distributions.

We consider different scenarios for the number of observations, $n \in (100, 200)$, the number of available predictors, $p \in (50, 100, 200)$, and error distribution $F$ (Gaussian, Laplace, Symmetric Stable with stability parameter $\alpha=1.5$).
Good leverage points are introduced multiplying the $(p-s)/2$ largest covariate values for 20\% of observations by 8.
These values are introduced in covariates with a true coefficient of 0 and hence should not affect the estimators.
Furthermore, 30\% of observations come from a different model with three distinct contamination signals.
We would expect that the PENSE objective function has at least one local minimum close to each of them.
The contamination signals all follow model~\ref{eqn:sim-true-model} but with a different $\mat\beta_0$ and further introducing leverage points in the truly relevant covariates.
The details are described in the Supplementary Materials Section~\ref{sec:suppl-sim-cont-dgp}.

For each of the 18 settings we repeat the simulation 50 times and compare the prediction performance of the solutions/penalty levels selected by N-CV and RIS-CV with $K=7$ folds and $R=5$ replications.
We select the penalization level such that the solution at this penalization level is within one standard error of the solution with smallest estimated prediction error (the 1-SE-rule).
As before, 40 solutions are retained for RIS-CV, and PENSE is tuned to a breakdown point of 40\% and $\alpha=0.5$.

Figure~\ref{fig:sim-pred-accuracy} summarizes the simulation results in terms of the prediction accuracy.
Here we show the difference in the true prediction error achieved with the solution chosen by RIS-CV and the solution chosen by N-CV.
The difference is scaled by the true scale of the error distribution, $F$, which can be different in each simulation run.
In the majority of simulation runs, RIS-CV selects solutions with smaller prediction error than N-CV.
While for Gaussian errors the gain from RIS-CV is sometimes negligible, for heavy-tailed errors distributions the average gains are more substantial.
Even in instances where the gains are not as pronounced, the smoother CV curves from RIS-CV nevertheless provide better guidance to the practitioner as to what penalization level should be chosen.
Section~\ref{sec:supp-additional-sim} of the supplementary materials includes additional results from the simulation study.

\begin{figure}[t]
  \centering
  \includegraphics[width=1\linewidth]{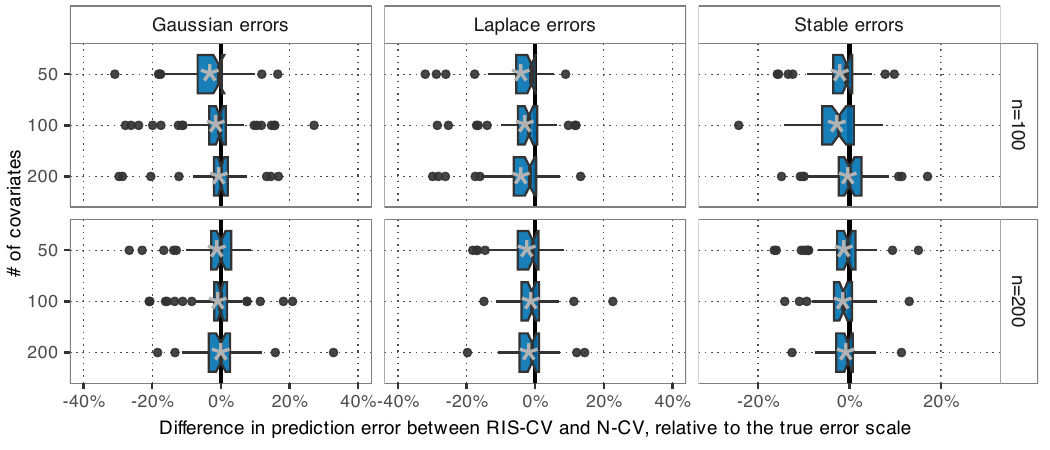}
  \caption{%
  Difference in prediction error between RIS-CV and N-CV.
  Penalty parameters are chosen by the 1-SE rule.
  The gray asterisks depict the mean.
  Negative differences mean the solution chosen by RIS-CV leads to better prediction accuracy than N-CV.%
  }
  \label{fig:sim-pred-accuracy}
\end{figure}

\section{Conclusion}

Cross-validation is the prevalent data-driven method to select models with penalized estimators.
We demonstrate that standard CV (N-CV) suffers from severe instability when applied to non-convex robust penalized regression estimators.
We demonstrate with theory and empirical results how the large number of local minima, caused for instance by outliers and contamination, can lead to highly non-smooth penalization paths and CV curves.
This non-smoothness can in turn lead to low-quality estimates of the prediction accuracy and thus an ill-guided selection of the penalization level.
These issues are typically more pronounced as the complexity of the data analysis task increases.
This inherent instability of N-CV has been a major obstacle for the utility and adoption of robust penalized estimators.

In this paper we therefore propose a novel strategy, RIS-CV, where we retain all local minima uncovered by the numerical optimization routine and share outlier information between these minima on the full data set with the individual cross-validation folds.
Our results show that by leveraging the robustness weights associated with each local minimum we can (a)~determine which minima in the CV folds correspond most closely with the minima on the full data set and (b)~estimate the prediction accuracy of all of those local minima, not only the global minimum.
This allows us to select the minimum with the best prediction accuracy, which is not necessarily the global minimum, yielding a smoother CV curve and penalization path.
RIS-CV further speeds up computations as the search space in the CV folds can be restricted to neighborhoods around the minima on the full data set.

The proposed matching scheme currently does not differentiate between good and poor matches.
The most similar CV solution is chosen as surrogate, irrespective of the actual similarity.
Since the metric is a unit-less correlation coefficient, thresholding rules could be developed in the future to avoid using unrelated minima in RIS-CV.
For example, one could require CV surrogates to have a similarity of at least 0.75.
Further research would be necessary, however, to devise appropriate strategies to handle situations where some CV folds do not yield a CV surrogate, and how to properly choose the threshold.

RIS-CV is a much-needed tool to improve the practicality, utility and acceptance of robust penalized estimators.
Our numerical studies reveal that RIS-CV leads to smoother CV curves and more reliable selection of the penalty parameter and, at the same time, the most suitable minimum of the objective function at the chosen penalization level.
We show that the improved smoothness and identification of useful minima leads to better out-of-sample prediction accuracy in a large-scale simulation study and in the real-world applications.
RIS-CV is thus improving the reliability of the robust model selection process and thereby instilling more trustworthiness in the results.

\printbibliography

\section{Supplementary Material}

\subsection{Failures of Naïve Cross-Validation}
\subsubsection{Non-smooth path of global minima}

Here we demonstrate that the chances of the global minimum ``jumping'' between local minima when the penalization level changes is non-negligible.
The following proposition shows that if there are two local minima (the ``good'' and the ``bad'' minimum), with the bad minimum being much closer to the origin than the good minimum, the bad minimum will take over from the good minimum as the global minimum when the penalization level is increased.
Under the stated conditions, the proposition is entirely deterministic.
We will show later that there are indeed situations in which the conditions for the proposition are satisfied with non-zero probability.

\begin{proposition}\label{prop:suppl-nonsmooth}
Consider a bounded robust loss function $\rho$ with $\psi(x) = \rho'(x) = 0$ for $|x| > c_1 > 0$ and $\psi'(x) \geq 0$ for $|x| < c_2 < c_1$ and $c_2 > 0$.
Assume that $\bhc$ and $\bhs$ are the only two minima of the objective function at a penalization level $\l$, that $\bhc > 0$, $\bhs - \bhc \gg 2 c_1$ and $\loss(\bhc; \l) = \loss(\bhs; \l)$.
Assume further that there exists a subset of the $n$ observations, $\mathcal C \subset {1, \dotsc, n}$, such that $|y_i - \bhc x_i| < c_1$ for all $i \in \mathcal C$, $|y_i - \bhs x_i| < c_1$ for all $i \notin \mathcal C$, and that the cardinality of the set ${i \colon x_i = 0}$ is less than $b n$.
This implies that $\bhc$ is determined only by observations in $\mathcal C$ and $\bhs$ is determined only by $i \notin \mathcal C$.

Then, for any $\delta$ with $|\delta|$ small enough, $\btc = \bhc - \delta/S_c$ with $S_c = \frac{1}{n} \sumc \psi'(y_i - (\bhc + \nu_c) x_i) x_i^2$ and $\nu_c \in (0, \delta)$ is a minimum of the objective function for penalization level $\l + \delta$.
Similarly, with $S_\star = \frac{1}{n} \sums \psi'(y_i - (\bhs + \nu_\star) x_i) x_i^2$ and $\nu_\star \in (0, \delta)$, $\bts = \bhs - \delta/S_\star$ is another minimum of the objective function.
Furthermore,
\begin{align*}
\loss(\bhc + |\delta|/S_c; \l - |\delta|) &> \loss(\bhs + |\delta|/S_\star; \l - |\delta|), \text{ and}\\
\loss(\bhc - |\delta|/S_c; \l + |\delta|) &< \loss(\bhs - |\delta|/S_\star; \l + |\delta|).
\end{align*}

Therefore, 
$$
\lim_{\tilde\delta \to 0} | \hat\beta(\l - \tilde\delta) - \hat\beta(\l + \tilde\delta) | > 0.
$$
In other words, under the above assumptions the regularization path for parameter $\beta$ has a discontinuity at $\l$ and hence is non-smooth.
\end{proposition}

\begin{proof}
The first step is to show that $\btc$ and $\bts$ are minima of $\loss(\beta; \l + \delta)$.
Since only $i \in \mathcal C$ have a non-zero derivative at $\bhc$ and $\bhs - \bhc \gg 2 c_1$ are far apart, the same is true for any small enough perturbation $\bhc + \eta$.

\begin{align*}
n (\lambda + \delta) &= \sumc \psi(y_i - (\bhc + \eta) x_i) x_i \\
\Leftrightarrow n (\lambda + \delta) &= \sumc \psi(y_i - \bhc x_i) x_i - \eta \psi'(y_i - (\bhc + \nu) x_i) x_i^2 \\
\Leftrightarrow n (\lambda + \delta) &= n \lambda - \eta \sumc \psi'(y_i - (\bhc + \nu) x_i) x_i^2 \\
\Leftrightarrow \eta &= -\frac{\delta}{\frac{1}{n} \sumc \psi'(y_i - (\bhc + \nu) x_i) x_i^2},
\end{align*}
where $\nu \in (0, \eta)$.
Therefore, we have $\btc = \bhc - \frac{\delta}{S_c}$ and if $\delta$ is small enough, $S_c > 0$.
Similarly, $\bts = \bhs - \frac{\delta}{S_\star}$ with $S_\star > 0$.

Next we need to show that $\loss(\btc; \l + \delta) < \loss(\bts; \l + \delta)$
A Taylor series expansion of $\loss(\bhc + \delta/S_c; \l + \delta)$ gives
$$
\loss(\bhc + \delta/S_c; \l + \delta) = \loss(\bhc; \l) + \delta \bhc - \frac{\delta^2}{S_c} + \frac{2 \delta^2}{S_c^2} \tilde S_c ,
$$
where $\tilde S_c = \sumc \psi'(y_i - x_i \xi) x_i^2 > 0$ and $\xi \in (0, \delta/S_c)$.
Applying a similar expansion for $\loss(\bhs + \delta/S_\star; \l + \delta)$ and noting that $\loss(\bhc; \l) = \loss(\bhs; \l)$ we get

\begin{equation}\label{eqn:suppl-nonsmooth-proof-gt0}
\begin{aligned}
\loss(\bhs + \frac{\delta}{S_\star}; \l + \delta) - \loss(\bhc + \frac{\delta}{S_c}; \l + \delta)
  &= \delta \bhs - \frac{\delta^2}{S_\star} + \frac{2 \delta^2}{S_\star^2} \tilde S_\star -
     \delta \bhc + \frac{\delta^2}{S_c} - \frac{2 \delta^2}{S_c^2} \tilde S_c \\
  &= \delta (\bhs - \bhc) + \delta^2\left [ \frac{1}{S_c} - \frac{1}{S_\star} + \frac{2 \tilde S_\star}{S_\star^2} - \frac{2 \tilde S_c}{S_c^2}  \right].
\end{aligned}
\end{equation}

Since $\rho$ is quadratic around 0 and bounded, and because less than $b n$ observations have $x_i = 0$, $S_c, S_\star, \tilde S_c, \tilde S_\star$ are all bounded and $S_c, S_\star$ are greater than 0 for $|\delta|$ small enough.
Therefore, $0 \leq \left| \frac{1}{S_c} - \frac{1}{S_\star} + \frac{2 \tilde S_\star}{S_\star^2} - \frac{2 \tilde S_c}{S_c^2} \right| < \infty$.
In turn, for $|\delta|$ small enough, \eqref{eqn:suppl-nonsmooth-proof-gt0} is strictly greater than 0 for positive $\delta$ and strictly less than 0 for negative $\delta$.

\end{proof}

\subsubsection{Example scenario}\label{sec:suppl-nonsmooth-example}

Here we consider a simple case where the conditions required for Proposition~\ref{prop:suppl-nonsmooth} hold with high probability.
Consider a bounded robust loss function $\rho$ such that $\rho(x) = \rho_\infty$ for all $|x| > c_1$ and $\rho$ is quadratic for $|x| \leq c_1$.
Examples of such loss functions are Hampel's loss \parencite{hampel_robust_1986} or the GGW and LQQ loss functions \parencite{koller_sharpening_2011}).
Other loss functions which are at least approximately quadratic in an open neighborhood of 0 would have similar behavior.

When we have $n$ independent realizations from the simple model
$$
y_i = \begin{cases}
  x_i \bc + \gamma_{c,i} & i \in \mathscr{C} \\
  x_i \bs + \gamma_{\star,i} & i \notin \mathscr{C},
\end{cases}
$$
where $\bc > 0$, $\bs - \bc$ large enough (see below), $x_i$ i.i.d.\ $N(0, 1)$, $\gamma_{c,i}$ and $\gamma_{\star,i}$ i.i.d.\ $N(0, \sigma_c^2)$ and $N(0, \sigma_\star^2)$, respectively, and $\mathscr{C} \subset \{1, \dotsc, n\}$ such that $|\mathscr C| = b n < n / 2$.
For simplicity we assume that $b n$ and hence $(1 - b) n$ are integer.
In this setting we can choose the parameters such that the conditions on $\bhc$ and $\bhs$ for Proposition~\ref{prop:suppl-nonsmooth} are satisfied with arbitrarily high probability.

In the following we will consider a ``bad'' minimum $\bhc \in \mathcal B_c = [\bc - \frac{\l n}{(b n - 2)} - \delta, \bc - \frac{\l n}{(b n - 2)} + \delta]$ and a ``good'' minimum $\bhs \in \mathcal B_\star = [\bs - \frac{\l n}{(n (1 - b) - 2)} - \delta, \bs - \frac{\l n}{(n (1 - b) - 2)} + \delta]$.
We assume that $\bc$ and $\bs$ are far enough apart such that the objective function $\loss(\bhc; \l)$ depends only on observations $i \in \mathscr C$ and $\loss(\bhs; \l)$ depends only on observations $i \notin \mathscr C$ with probability at least $1 - \kappa$.
In other words, the residuals for observations in $\mathcal C$ are within the quadratic part of the loss function for any $\bhc$ and in the bounded region for any $\bhs$, and vice versa for observations not in $\mathcal C$.

\begin{proof}
Since we assume that the robust loss function $\rho$ is bounded, we need to show that for all $\bhc \in \mathcal B_c$ all residuals for observations in $\mathcal c$ are less than $c_2$ in absolute value with high probability.
Writing $r_{c,i} = y_i - \bhc x_i$ we have that $r_{c,i} = \gamma_{c,i} + \frac{\l n}{(b n - 2)} x_i + \delta_c$ for all $i \in \mathcal{C}$ and hence
\begin{align*}
\P{|r_{c,i}| < c_2 } 
  & = \P{|\gamma_{ci} + \frac{\l n}{(b n - 2)} x_i + \delta_c| < c_2 } \\
  & = 1 - 2 \Phi\left(- \frac{c_2}{\sqrt{\sigma_c^2 + \frac{\l^2 n^2}{(b n - 2)^2} + \delta_c^2 }} \right) \\
  & = p_{c, \mathcal C} > 0.
\end{align*}
Similarly, for $i \notin \mathcal{C}$,
\begin{align*}
\P{|r_{c,i}| > c_1 } 
  & = \P{|\gamma_{ci} + (\bs - \bc) x_i + \frac{\l n}{(b n - 2)} x_i + \delta_c| > c_1 } \\
  & = 2 \Phi\left(- \frac{c_1}{\sqrt{\left(\sigma_c^2 + \frac{\l^2 n^2}{(b n - 2)^2} + \delta_c^2 \right) + (\bs - \bc)^2 }} \right) \\
  & = p_{c, \mathcal C^c} \gg 1 - p_{c, \mathcal C}.
\end{align*}
The probability $p_{c, \mathcal C^c}$ can be made arbitrarily large by moving $\bc$ and $\bs$ arbitrarily far apart.
Since $r_i$ are i.i.d., $\P{\forall i \in \mathcal{C}\colon |r_{c,i}| \leq c_2 \wedge \forall i \notin \mathcal{C}\colon |r_{c,i}| > c_1} = p_{c, \mathcal C}^{bn} p_{c, \mathcal C^c}^{(1-b)n} \gg 0$.

The same calculations can be done for $\bhs \in \mathcal B_\star$, and hence the objective function value at $\bhs$ and $\bhc$ do not depend on the same observations with arbitrarily high probability.
\end{proof}

Conditioned on the partitioning of the observations from above, the objective function has at least one minimum in each of $\mathcal B_c$ and $\mathcal B_\star$ with probability at least $1 - \kappa$, i.e.,

\begin{equation}\label{eqn:supp-ex1-step2-prob}
\P{\exists (\bhc, \bhs) \in \mathcal B_c \otimes \mathcal B_\star \colon \loss'(\bhc; \l) = \loss'(\bhs; \l) = 0 } > 1 - \kappa.
\end{equation}

\begin{proof}
Consider $\bhc= \bc - \frac{\l n}{(b n - 2)} + \Delta$.
For $\bhc$ to be a minimum, the derivative of the penalized loss must be 0.
As we assume that $\bhc > 0$, this is equivalent to:

\begin{equation}\label{eqn:supp-ex1-step2-deriv0}
n\l = \sumc \psi \left(y_i - \bc x_i - \frac{\l n}{(b n - 2)}  x_i + \Delta  x_i \right) x_i.
\end{equation}

Since $\rho$ is quadratic at 0 $\psi$ is linear and hence, for $\sigma_c^2$ and $\l$ small enough, we can re-write \eqref{eqn:supp-ex1-step2-deriv0} as
\begin{align*}
    n\l &= \sumc \gamma_i x_i + \frac{\l n}{(b n - 2)} \sumc x_i^2 + \Delta \sumc x_i^2 \\
    \Leftrightarrow
    \Delta &= \frac{n\l (-1 + \frac{1}{(b n - 2)} \sumc x_i^2) + \sumc \gamma_i x_i}
                   {\sumc x_i^2}.
\end{align*}

Therefore, $\E[\Delta] = 0$ and $\V[\Delta] \leq \frac{1}{(b n - 2)} \left[\sigma_c^2 + \frac{2 \l^2 n^2}{(bn-4)(bn-2)} \right]$.
Similar calculations can be carried out for $\bhs$.
For any given $\epsilon, \kappa > 0$ we can therefore find suitable $\sigma_c^2$, $\sigma_\star^2$ and $n$ to satisfy \eqref{eqn:supp-ex1-step2-prob}.
\end{proof}

Furthermore, there exists a sequence $\l_n$ such that the objective function values at $\bhc$ and $\bhs$ are within an $\epsilon$ neighborhood with arbitrarily high probability $1 - \kappa$, i.e., 

\begin{equation}\label{eqn:supp-ex1-step3-prob-limit}
\lim_{n \to \infty} \P{\loss(\bhs; \l_n) = \loss(\bhc; \l_n) } > 1 - \kappa.
\end{equation}

\begin{proof}
Define $D = \loss(\bhs; \l_n) - \loss(\bhc; \l_n)$.
Since the loss function is quadratic in a neighborhood around 0, we can write $D$ as

\begin{align*}
D =&\, \rho_\infty (b - (1 - b)) + \\
  &\frac{1}{n} \left(\sums \gamma_i^2 - \sumc \gamma_i^2 \right) + \\
  &\l_n^2 \left(\frac{n}{(n(1 - b) - 2)^2} \sums x_i^2 - \frac{n}{(n b - 2)^2} \sumc x_i^2 \right) + \\
  &2 \l_n \left(\frac{1}{(n(1 - b) - 2)} \sums \gamma_i x_i - \frac{1}{(n b - 2)} \sumc \gamma_i x_i \right) + \\
  &2 \l_n \left(\frac{\Delta_\star}{(n(1 - b) - 2)} \sums x_i - \frac{\Delta_c}{(n b - 2)} \sumc x_i \right) + \\
  &\frac{1}{n} \left(\Delta_\star^2 \sums x^2 - \Delta_c^2 \sumc x_i^2 \right) + \\
  &\frac{2}{n} \left(\Delta_\star \sums \gamma_i x_i - \Delta_c \sumc \gamma_i x_i \right) + \\
  &\l_n^2 \left(-\frac{n}{n(1-b) - 2} + \frac{n}{nb - 2} \right) \\
  &\l_n \left(\bs-\bc + \Delta_\star - \Delta_c \right).
\end{align*}

Setting the expectation of $D$ to 0 yields
\begin{align*}
0 = \E[D] =&\, \rho_\infty (2 b - 1) + \\
  &(1 - b) \sigma_\star^2 - b \sigma_c^2 + \\
  &\l_n^2 \left( \frac{n^2 (1 - b)}{(n(1 - b) - 2)^2} - \frac{n^2 b}{(n b - 2)^2}\right) + \\
  &\Delta_\star^2 (1 - b) - \Delta_c^2 \sigma_c^2 b + \\
  &\l_n^2 \left(-\frac{n}{(n(1-b) - 2)} + \frac{n}{(nb - 2)} \right) \\
  &\l_n \left(\bs-\bc + \Delta_\star - \Delta_c \right) \\
=&\,
  \l_n^2 \left[ \frac{n^2 (1 - b)}{(n(1 - b) - 2)^2} - \frac{n}{(n(1-b) - 2)} - \frac{n^2 b}{(n b - 2)^2} + \frac{n}{(nb - 2) } \right] + \\
  & \l_n \left(\bs-\bc + \Delta_\star - \Delta_c \right) + \\
  & \rho_\infty (2b - 1) + (1 - b) \sigma_\star^2 - b \sigma_c^2.
\end{align*}

Setting $A_n = \left[ \frac{n^2 (1 - b)}{(n(1 - b) - 2)^2} - \frac{n}{(n(1-b) - 2)} - \frac{n^2 b}{(n b - 2)^2} + \frac{n}{(nb - 2) } \right]$, $B = \bs-\bc + \Delta_\star - \Delta_c$ and $C = \rho(\infty) (b - (1 - b)) + (1 - b) \sigma_\star^2 - b \sigma_c^2$, we can see that $A_n < 0$ with $\lim_{n \to \infty} A_n = 0$ and $B > 0$.
Further, we can choose $\sigma_c$ and $\sigma_\star$ such that $C < 0$.
Now if $\bs - \bc = o((\sigma_\star^2 - \sigma_c^2)/n)$ then $B^2 - 4 A C > 0$ and hence there exists a $\l_n > 0$ such that $\E[D] = 0$.
Specifically,

$$
\lim_{n \to \infty} \l_n = \frac{b \sigma_c^2 - (1 - b) \sigma_\star^2  + \rho_\infty (1 - 2 b)} {\bs - \bc + \Delta_\star - \Delta_c}.
$$

Moreover, $\V[D] = \frac{\sqrt{2} b}{n} \left(\sigma_c^2 + C_{cn}^2 \right) + \frac{\sqrt{2} (1 - b)}{n} \left(\sigma_\star^2 + C_{\star n}^2 \right)$ with $C_{cn} = \frac{\l_n n} {(n b - 2)} + \Delta_c$ and $C_{\star n} = \frac{\l_n n} {(n (1 - b) - 2)} + \Delta_\star$ which goes to 0 as $n$ goes to infinity.
Therefore, $\lim_{n\to\infty}\P{|D| > \epsilon} = 0.$
\end{proof}

\begin{figure}
    \centering
    \includegraphics{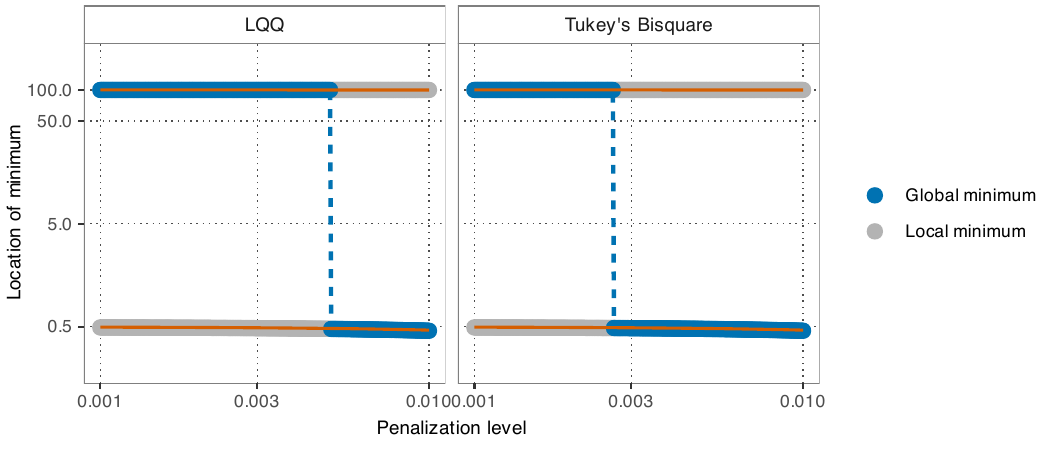}
    \caption{Demonstration of a non-smooth regularization path for a penalized M-estimator of regression in a simulation following the example scenario from Section~\ref{sec:suppl-nonsmooth-example} (with $\sigma_c=0.01$, $\sigma_\star = 0.1$, $\bc = 0.5$, $\bs = 100$, $b = 0.3$ and $n = 100$).
    We show the location of local minima when using the LQQ $\rho$ function (left panel) and Tukey's bisquare $\rho$ function (right panel).
    Gray dots represent local minima and blue dots indicate the global minimum.
    The orange lines depict the expected value of the minima at $\bhc = \bc - \frac{\l n}{n b - 2}$ and $\bhs = \bs - \frac{\l n}{n (1 - b) - 2}$.}
    \label{fig:suppl-nonsmooth-demo}
\end{figure}

\subsection{Additional Empirical Results}\label{sec:suppl-additional-empirical-results}

\subsubsection{Real-world Application II: Gene Pathway Recovery}\label{sec:appl-pw}

The second application is a gene pathway recovery analysis using data from \textcite{pfister_stabilizing_2021}.
The data set contains preprocessed protein expression levels from 340 genes from seven different pathways for 315 subjects.
Following the original analysis \parencite{pfister_stabilizing_2021} we define as response the average expression of proteins on the \textit{Cholesterol Biosynthesis} pathway.
We further add a Laplace-distributed noise to achieve a signal-to-noise ratio (SNR) of 1.
We split the data set into a training data set comprising 165 randomly selected subjects and a test data set with the remaining 150 subjects.
We further contaminate the training data set by replacing the response for 25 subjects (15\%) with the average expression of proteins on the \textit{Ribosome} pathway, again adding Laplace noise with a SNR of 1.
The PENSE estimator is tuned to a breakdown point of 25\% and uses an EN penalty with $\alpha=0.5$.
For RIS-CV we retain up to 40 local solutions.

In contrast to the previous application, here we can evaluate the actual prediction performance evaluated on an independent test set.
In Figure~\ref{fig:appl-pw-cv}a we see both the CV estimated prediction errors and the true RMSPE evaluated on the independent test set.
In this application N-CV is not as affected by local optima than in the previous example.
However, there is still a discontinuity around $\lambda\approx10^{-1}$ for N-CV.
RIS-CV, on the other hand, again yields a smoother CV curve and hence a more reliable selection of the penalization level.
The actual prediction error of the RIS-CV solution is about 3\% lower than that of the solution selected by N-CV, and computations are also roughly 3\% faster.

Both the true RMSPE for PENSE in Figure~\ref{fig:appl-pw-cv}a and the $L_1$ norm of the selected solution in Figure~\ref{fig:appl-pw-cv}b again highlight the benefits of considering more than just the global minimum.
With the ability of RIS-CV to estimate the prediction error of all local minima we can select a more appropriate minimum than with N-CV.
The instability in the true RMSPE from N-CV is indicative of the global minimum alternating between two or more distinct local minima.
This is also visible in the the $L_1$ norm of the chosen solutions in Figure~\ref{fig:appl-pw-cv}b, showing substantial instability in the global minimum.
While there is also some instability in the penalization path of RIS-CV, it seems less severe and further from the penalization levels of interest.
For convex objective functions the $L_1$ norm of the EN penalty is a monotone function of the penalization level (± small deviations due to the addition of the $L_2$ penalty in the EN formulation).
For the non-convex PENSE objective function this is not necessarily the case, as evident in the plot.
For the minimum selected by RIS-CV, however, the monotonicity mostly holds in the region of interest.

Figure~\ref{fig:appl-pw-ind} shows the CV curves from two independent 10-fold CV runs, one set estimated by naïve CV, the other by RIS-CV.
As seen in the first application, the estimated prediction error as a function of the penalization level is much less smooth for naïve CV than for RIS CV.

\begin{figure}[t]
  \centering
  \includegraphics[width=1\linewidth]{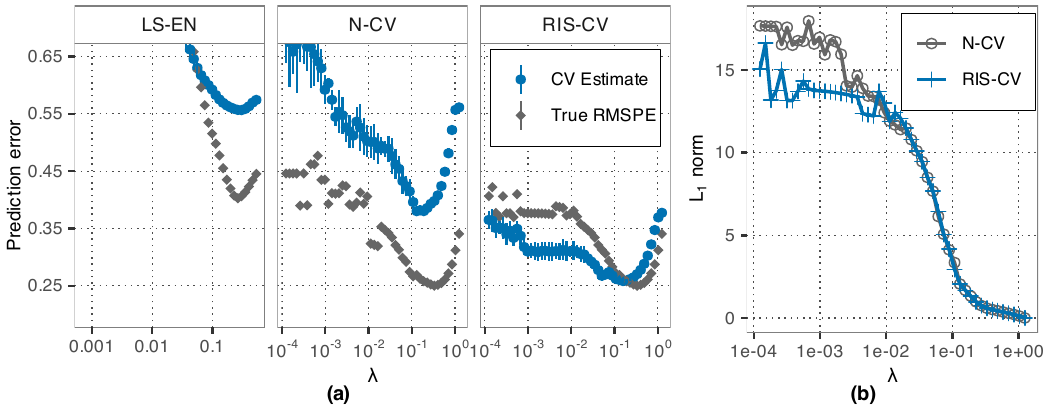}
  \caption{%
  Results from the pathway recovery analysis:
  (a) CV curves from five replications of 10-fold CV (blue) alongside the true RMSPE of the fitted models applied to the independent test set (gray), and
  (b) the $L_1$ norm (excluding the intercept) of the selected minimum of the PENSE objective function versus the penalization level $\lambda$.
  The left panel in (a) shows the CV estimated mean absolute prediction error and the true RMSPE for classical EN (``LS-EN'').
  The center panel in (a) shows the naïve CV estimated $\tau$-size of the prediction error and the true RMSPE of PENSE (``N-CV'').
  The right panel in (a) shows the RIS-CV estimated weighted RMSPE of PENSE alongside the true RMSPE (``RIS-CV'').%
  }
  \label{fig:appl-pw-cv}
\end{figure}

\begin{figure}[ht]
  \centering
  \includegraphics[width=1\linewidth]{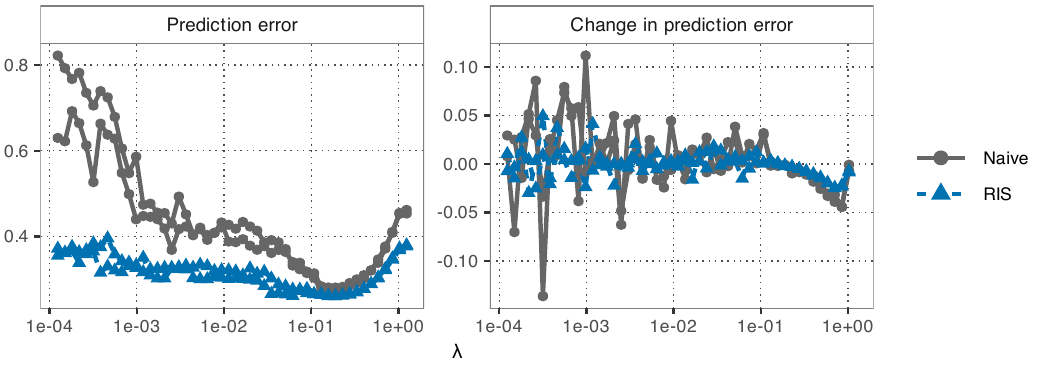}
  \caption{%
  Estimated prediction errors (left) and the respective changes (right) from two independent 10-fold CV runs for PENSE applied to the gene pathway recovery analysis.
  The blue curves represent the estimated weighted RMSPE from RIS-CV, while the gray curves show the $\tau$-size of the prediction errors estimated by naïve CV.
  The estimated prediction error for N-CV is shifted downwards by 0.1 for easier comparison with RIS-CV.%
  }
  \label{fig:appl-pw-ind}
\end{figure}

\subsubsection{Real-world Application III: Determinants of Plasma Beta-Carotene Levels}\label{sec:appl-plasma}

In this additional application we try to determine determinants of plasma beta-carotene levels using publicly available data obtained from \url{http://lib.stat.cmu.edu/datasets/Plasma_Retinol} \parencite{nierenberg_determinants_1989}.
We dummy-code the data and build a model with all available covariates and their interaction with the subject's sex (binary male/female).
This leads to a total of 22 predictors for 315 subjects (42 male, 273 female).
We randomly select a training data set of size $N=100$, stratified among male/female such that 30 subjects in the training data are male and 70 are female.
We deliberately oversample male subjects to ensure sufficient variation in all sex-dependent interaction terms.

In the individual CV runs shown in Figure~\ref{fig:appl-plasma-ind} it is clear that 10-fold naïve CV again is much less stable than RIS-CV.
Particularly for penalty levels of interest N-CV has difficulty estimating the prediction accuracy.
This is also visible when averaging over five replications of 10-fold CV.
N-CV still exhibits substantially more variation and non-smoothness for penalization levels of interest ($\lambda \in [1, 10]$) than RIS-CV.
Here it is obvious that a major driver of the variability in the CV estimated prediction error is not the Monte Carlo error in the CV splits, but the the non-convexity of the objective error.

\begin{figure}[ht]
  \centering
  \includegraphics[width=1\linewidth]{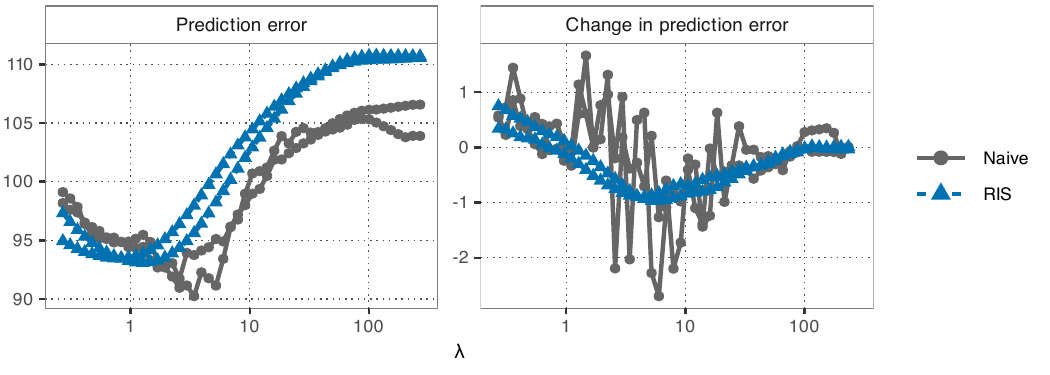}
  \caption{%
  Estimated prediction errors (left) and the respective changes (right) from two independent 10-fold CV runs for PENSE applied to the analysis of plasma beta-carotene levels.
  The blue curves represent the estimated weighted RMSPE from RIS-CV, while the gray curves show the $\tau$-size of the prediction errors estimated by naïve CV.
  The estimated prediction error for RIS-CV is shifted upwards by 25 for easier comparison with N-CV.%
  }
  \label{fig:appl-plasma-ind}
\end{figure}

\begin{figure}[t]
  \centering
  \includegraphics[width=1\linewidth]{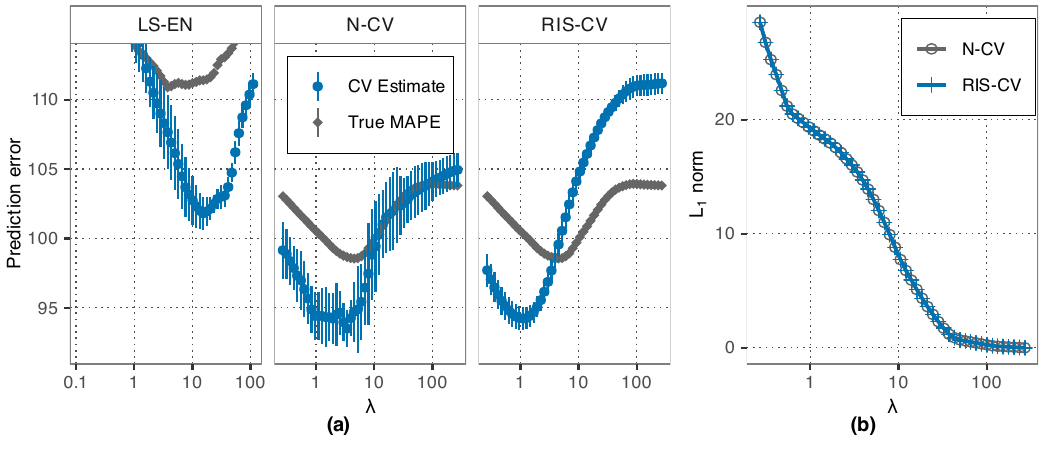}
  \caption{%
  CV curves from five replications of 10-fold CV (blue) for the analysis of plasma beta-carotene levels, alongside the true MAPE of the fitted models applied to the independent test set (gray).
  The left panel (``LS-EN'') shows the CV estimated mean absolute prediction error and the true MAPE for classical EN.
  The center panel (``N-CV'') shows the naïve CV estimated $\tau$-size of the prediction error and the true MAPE of PENSE.
  The right panel (``RIS-CV'') shows the RIS-CV estimated weighted RMSPE of PENSE (shifted upwards by 25) alongside the true MAPE.%
  }
  \label{fig:appl-plasma-cv}
\end{figure}

\subsubsection{Details About the Simulation Settings}\label{sec:suppl-sim-cont-dgp}

The contamination data generating process is defined as follows.
For each contamination signal we first randomly select $\lfloor \log_2(p) \rfloor$ covariates (excluding the first $s$ covariates), denoted by $\mathcal J^* \subset \{s + 1, \dotsc, p\}$.
Then, for three different values of $u_1 = -1.5, u_2 = -1, u_3 = -0.5$ and the respective contamination indices $\mathcal C_1 = \{1, \dotsc, 0.1 n \}$, $\mathcal C_2 = \{0.1 n + 1, \dotsc, 0.2 n \}$, $\mathcal C_3 = \{0.2 n + 1, \dotsc, 0.3 n \}$, in observations $i \in \mathcal C_k$ the covariates and responses are replaced according to the following DGP:
\begin{align*}
x^*_{ij} = \begin{cases}
    x_{ij} & j \notin \mathcal J^* \\
    k_l x_{ij} & j \in \mathcal J^* \\
\end{cases},
&&
\beta^*_{0j} = \begin{cases}
    0 & j \notin \mathcal J^* \\
    k_v & j \in \mathcal J^* \\
\end{cases},
&&
y^*_{i} = {\mat x^*_i}\tr \mat\beta^*_0 + \varepsilon^*_i.
\end{align*}
The constant $k_l$ is chosen such that $\mat x^*_i$ is at least twice as far from the center (in terms of the Mahalanobis distance) than all the other non-contaminated observations.
The error term $\varepsilon^*_i$ is Gaussian with variance such that $\mat\beta^*_0$ achieves a SNR of 10.

\subsubsection{Additional Simulation Results}\label{sec:supp-additional-sim}

Here we present additional results from the simulation study.
Figure~\ref{fig:suppl-sim-examples} shows the CV curves and the corresponding true prediction errors for simulation runs from two different settings.
In the left panel it is evident that for heavy tailed errors RIS-CV yields both better prediction accuracy and also a smoother CV curve.
While for Gaussian errors the difference in prediction performance is marginal, a practitioner may have a difficult time selecting a suitable penalization level because the CV curve does not exhibit the usual ``U'' shape.

\begin{figure}[t]
  \centering
  \includegraphics[width=1\linewidth]{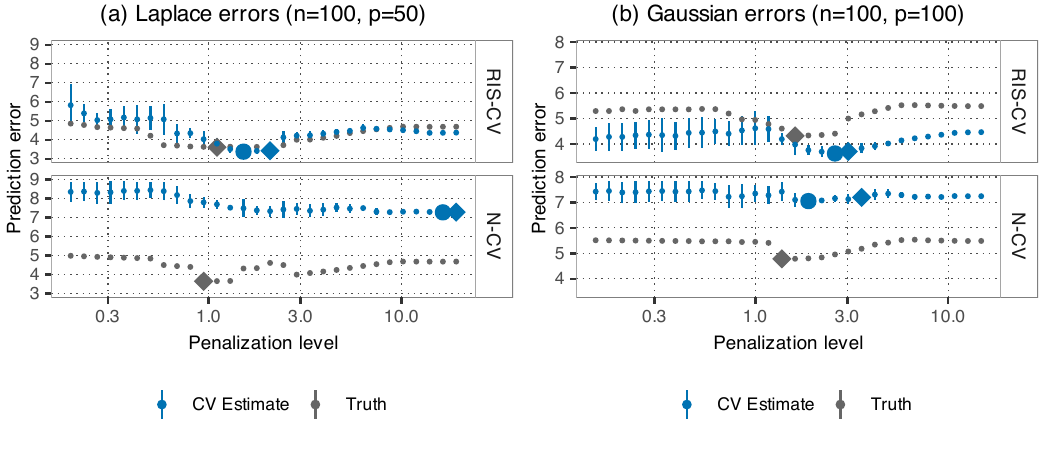}
  \caption{%
  Examples of CV curves estimated by RIS-CV (top) and N-CV (bottom). For CV estimates (blue curves with error bars), the large dots depict the solution with smallest prediction error, and the diamonds depict the solutions within 1 standard error. For the true prediction errors (gray curves without error bars), the diamond depicts the optimal solution.%
  }
  \label{fig:suppl-sim-examples}
\end{figure}

\end{document}